\documentclass[10pt, a4paper]{article}

\usepackage{amssymb,amsmath,mathtools} 
\usepackage{graphicx} 
\newcommand{\bm}{}

\usepackage[hmargin=0.12\paperwidth,vmargin=0.2\paperwidth,bindingoffset=0cm]{geometry}

\usepackage{titlesec}

\titleformat*{\section}{\large\bfseries}
\titleformat*{\subsection}{\bfseries}

\usepackage{amsthm}

\theoremstyle{plain}
  \newtheorem{theorem}{Theorem}[section]

  \newtheorem{corollary}[theorem]{Corollary}
\theoremstyle{definition}
  
  \newtheorem{remark}[theorem]{Remark}
  
  \newtheorem*{notation}{Notation}
\theoremstyle{remark}

\newcommand{\p}{\partial} 

\newcommand{\nn}{\nonumber} 

\renewcommand{\vec}{\mathbf} 

\newcommand{\R}{\mathbb{R}}

\newcommand{\E}{\mathbb E}

\renewcommand{\phi}{\varphi} 
\renewcommand{\epsilon}{\varepsilon} 

\DeclarePairedDelimiter{\abs}{\lvert}{\rvert} 
\DeclarePairedDelimiter{\norm}{\lVert}{\rVert} 

\newcommand{\law}{\textup{law}}

\usepackage{hyperref}   
\hypersetup{
 pdfborder={0 0 0},      
 colorlinks=true,        
 linktoc=page,           
 linkcolor=blue,         
 citecolor=blue,         
}

\title{\bfseries\Large May--Wigner transition in large random dynamical systems}
\author{J.~R.~Ipsen\\[1em]%
\small ARC Centre of Excellence for Mathematical and Statistical Frontiers\\%
\small School of Mathematics and Statistics, The University of Melbourne, Victoria 3010, Australia%
}
\date{\today}

\begin{document}	

\maketitle

\begin{abstract}
\noindent
We consider stability in a class of random non-linear dynamical systems characterised by a relaxation rate together with a Gaussian random vector field which is white-in-time and spatial homogeneous and isotropic. We will show that in the limit of large dimension there is a stability-complexity phase transition analogue to the so-called May--Wigner transition known from linear models. Our approach uses an explicit derivation of a stochastic description of the finite-time Lyapunov exponents. These exponents are given as a system of coupled Brownian motions with hyperbolic repulsion called geometric Dyson Brownian motions. We compare our results with known models from the literature.
\end{abstract}

\section{Introduction}

In many complex systems, time evolution is modelled mathematically by a system of $n$ coupled non-linear first-order ordinary differential equations. In vector notation, we may write such a system as
\begin{equation}\label{dyn-sys}
\frac{d\vec x(t)}{dt}=\vec v(\vec x(t);t),
\end{equation}
where $\vec x$ is an $n$-dimensional vector-valued function and  $\vec v$ is an $n$-dimensional non-linear vector field specified by the model in question. As a mental image, we might think of $\vec x$ as the location of a test particle moving in a velocity field $\vec v$.

While a few properties of the dynamical system~\eqref{dyn-sys} can be answered within a fairly general context, most questions of practical interest are highly dependent on the choice of the vector field $\vec v$. This dependence appears to be somewhat problematic because many (if not most) large and complex systems seem inaccessible (even numerically) and, even worse, our knowledge about the vector field $\vec v$ is often incomplete. So how can we make in any practical predictions whatsoever? Inspired by statistical physics, it seems natural to take the vector field $\vec v$ to be `maximally' random under some physical constraints. The hope is that universal behaviour is present in the `thermodynamical' limit $n\to\infty$.

A celebrated and highly influential model was introduced in a 1972 paper by May~\cite{May:1972}. He imagined an autonomous dynamical system expanded to first order about an equilibrium point $\vec x_e$ and wrote
\begin{equation}\label{May}
\frac{d\vec r(t)}{dt}=-\mu\vec r(t)+\mathbf M\,\vec r(t)
\end{equation}
with $\vec r(t)=\vec x(t)-\vec x_e$ denoting a small perturbation, $\mu$ a positive constant and $\mathbf M$ an $n$-by-$n$ random matrix. This linear equation is, of course, much simpler than the original non-linear monstrosity~\eqref{dyn-sys}.
In~\eqref{May}, we should think of $\mu$ as a decay rate such that if the matrix $\mathbf M$ vanishes, then any small perturbation dies out at that rate. The matrix $\mathbf M$ incorporates local structure near the equilibrium and should be taken random up to some physical constraints. May was mainly interested in the stability of large ecosystems with $\mathbf M$ being a community matrix incorporating interactions between species (whether cooperative, competitive, or predator-prey interaction); we refer to~\cite{AT:2015} for a review of May inspired models in ecology. For many physicists, it might seem more natural to impose a symmetry constraint such as isotropy (in distribution). If we additionally assume Gaussian disorder, then this symmetry constraint would imply that $\mathbf M$ is a matrix from so-called Gaussian elliptic ensemble~\cite{SCSS:1988}. 

It is well-known that the linear model~\eqref{May} describes a stable (unstable) equilibrium point if the largest Lyapunov exponent is less (greater) than zero, or equivalently, if the eigenvalue of $\mathbf M$ with the greatest real part is less (greater) than $\mu$. The crucial observation made by May~\cite{May:1972} was that (under quite general conditions) this stability-complexity transition turns into a phase transition in the limit of large systems, $n\to\infty$. This phase transition is known as the May--Wigner transition. The simplest situation arises when $\mathbf M$ is a symmetric matrix in which case the eigenvalues of $\mathbf M$ is identical to the Lyapunov exponents. Here, large deviation results for symmetric Gaussian matrices has been used to show that the May--Wigner transition is a third order phase transition, see e.g.~\cite{MS:2014}.

The main drawback of May's model~\eqref{May} is that it is only physically meaningful locally. We can use May's model to determine whether a given equilibrium point is stable, but the model can never tell us whether this equilibrium is globally or locally stable, neither can it tell us the number of equilibria. This is a significant problem, because large complex systems are notorious for having a plethora of metastable states. In order to address this issue, Fyodorov and Khoruzhenko~\cite{FK:2016} have recently suggested a non-linear generalisation of May's model which is globally meaningful. The Fyodorov--Khoruzhenko model reads
\begin{equation}\label{FK-model}
\frac{d\vec x(t)}{dt}=-\mu\vec x(t)+\vec f(\vec x(t)),
\end{equation}
where $\mu$ is a positive constant (as in May's model) and $\vec f$ is an $n$-dimensional zero-mean Gaussian vector field assumed to be homogeneous and isotropic (in distribution). It was shown in~\cite{FK:2016} that this model also has a phase transition in the limit of large systems, $n\to\infty$ (actually there is an additional phase transition~\cite{BAFK:2017} but this additional transition is irrelevant for our purposes). If the fluctuations of the Gaussian field are sufficiently small then the system~\eqref{FK-model} is expected to have a single equilibrium point (this equilibrium must be stable since $\mu>0$); on the other hand, if the fluctuations of the Gaussian field are too large then the system is expected to have exponentially many equilibria. The region of phase space with a single equilibrium is interpreted as a stable phase, while the region with exponentially many equilibria is interpreted as a chaotic (or unstable) phase.
Thus, the phase transition between these two regions of phase space is considered as a non-linear generalisation of the May--Wigner transition.
We note that the Fyodorov--Khoruzhenko model is highly inspired by previous studies of random landscapes~\cite{Fyodorov:2004,FN:2012} (see~\cite{Fyodorov:2013} for a review) and it is intimately related to (mean field) spin glass theory~\cite{MPV:1987,CC:2005,AAC:2013}.

In both May's model~\eqref{May} and the Fyodorov--Khoruzhenko model~\eqref{FK-model} it is assumed that the dynamical system is autonomous, i.e. the driving vector field has no explicit dependence on time, $\vec v(\vec x;t)=\vec v(\vec x)$. This assumption is reasonable in some models but not in others. Moreover, even in models where the assumption of time-independence is reasonable, this is often only true in an approximate sense. For example, in ecosystems we would expect that the inter-species interactions change over time due to environmental changes. Likewise, in glassy systems we would expect a time-dependence to be introduced by atomic fluctuations. Thus, it seems more than reasonable to attempt to study dynamical systems~\eqref{dyn-sys} with an explicit time-dependence and, in particularly, provide non-autonomous generalisations of May's model and the Fyodorov--Khoruzhenko model. However, it is extremely challenging to include time correlations in a general setting. Thus, some simplification is needed also in the non-autonomous setting. A natural starting point for an analytic treatment is to consider systems driven by a vector field which is white-in-time. From a physical perspective, the assumption that our system is autonomous says that temporal fluctuations are slow compared to spatial fluctuations, while the white-in-time assumption says that spatial fluctuations are slow compared to temporal fluctuations. 

It is, of course, not a new idea to consider dynamical systems~\eqref{dyn-sys} in which the vector field $\vec v$ is white-in-time. In fact, this type of dynamical systems plays a central role in the study of fully developed turbulence~\cite{MY:1975,FGV:2001}, which was the inspiration for the mathematical theory of stochastic flows~\cite{Kunita:1997}. Surprisingly, these ideas have only very recently~\cite{IS:2016} been applied to the study of May--Wigner-like transitions. In~\cite{IS:2016}, we considered a white-in-time version of May's model~\eqref{May},
\begin{equation}\label{linear}
\frac{d\vec r(t)}{dt}=-\mu\vec r(t)+\mathbf M(t)\vec r(t),
\end{equation}
where $\mathbf M(t)$ is an isotropic Gaussian white-in-time process. It was shown that it was possible to switch from the matrix formulation~\eqref{linear} to a stochastic description of the finite-time Lyapunov exponents. This latter description was used to study stability. Interestingly, it was observed that the stability-complexity phase transition happens at a scale proportional to $n$, while the phase transition in May's model~\eqref{May} happens at a scale proportional to $n^{1/2}$. This is a manifestation of the (not so surprising) fact that the non-autonomous model~\eqref{linear} has a lower onset of complexity, than May's autonomous model~\eqref{May}.

Naturally, the linear model~\eqref{linear} has the same drawback as May's model, namely it is only physically meaningful in a neighbourhood of a stable trajectory. For this reason, it is the purpose of this paper to introduce and study stability in a white-in-time version of the Fyodorov--Khoruzhenko model which is globally meaningful. We will see that this new model has certain similarities with the linear model~\eqref{linear} but also some important differences. More precisely, we obtain a stochastic description for the finite-time Lyapunov exponents, which turns out to be identical to the linear model from~\cite{IS:2016} but with a restriction on parameter space. We use this description to draw a stability-complexity phase diagram. On a more heuristic level, we will also consider models with a more general time-dependence and discuss the possibility of persistent time-correlations.
The rest of this paper is organised as follows:\\
--- Section~\ref{sec:model}: we define the model which is the main subject of this paper.\\
--- Section~\ref{sec:from-DS-to-gDBm}: we show how to obtain a stochastic description of the finite-time Lyapunov exponents.\\
--- Section~\ref{sec:stability}: we study stability and observe that special behaviour arises in the weak-noise limit.\\
--- Section~\ref{sec:time}: we discuss the effect of time-dependence and auto-correlations in a more general setting.\\
--- Section~\ref{sec:conclusion}: we summarise our conclusions and state some open problems.\\
--- Appendix~\ref{appendix}: we present some important results for geometric Dyson Brownian motions.

\begin{notation}
Throughout this paper, lower case boldface symbols $\vec a,\ldots$ refer to $n$-dimensional (column) vectors, upper case boldface symbols $\mathbf A,\ldots$ refer to $n$-by-$n$ matrices, and $\vec a\bm\cdot\vec b$ refers to the usual dot product. Moreover, we will always use the standard notation,
\begin{equation}
\vec a=\begin{bmatrix} a_1 \\ \vdots \\ a_n\end{bmatrix}
\qquad \text{and} \qquad
\mathbf A=\begin{bmatrix} A_{11} & \cdots & A_{1n} \\ \vdots && \vdots \\ A_{n1} & \cdots & A_{nn} \end{bmatrix},
\end{equation}
for entries of vectors and matrices, and $\norm{\vec a}=(\vec a\bm\cdot\vec a)^{1/2}$ for the (Euclidean) norm.

\end{notation}

\section{A model for large complex dynamical systems}
\label{sec:model}

Let $\mu>0$ be a positive constant, $\vec x$ be an $n$-dimensional vector-valued function dependent on an external time $t$, and $\vec f$ be an $n$-dimensional vector field depending explicitly on the time $t$. Consider a generic dynamical system
\begin{equation}\label{SDE}
\frac{d\vec x(t)}{dt}=-\mu\vec x(t)+\vec f(\vec x(t);t)
\end{equation}
or equivalently, in integral form, 
\begin{equation}\label{SDE-int}
\vec x(t)-\vec x(s)=-\mu \int_s^t\vec x(u)du+\int_s^t \vec f(\vec x(u);u)du.
\end{equation}
In this paper, we consider the case where $\vec f$ is a zero-mean Gaussian random vector field, which is smooth-in-space and white-in-time. We note that this system is not independent of stochastic regularisation, i.e. there will be an `It\^o--Stratonovich dilemma'.
In the following, we have chosen the Stratonovich interpretation, which corresponds to the assumption that the underlying physical phenomenon responsible for the noise is time-reversible. We emphasize that all computations in this paper can easily be performed using any other stochastic regularisation (e.g. It\^o); this would merely introduce a spurious drift to the Lyapunov spectrum. 

Let us consider our random vector field in greater detail. Since $\vec f$ is zero-mean Gaussian, it is uniquely determined by its covariance (or diffusion) tensor
\begin{equation}\label{covar}
\E[f_i(\vec x;t)f_j(\vec y,s)]=D_{ij}(\vec x,\vec y)\delta(t-s),\qquad i,j=1,\ldots,n,
\end{equation}
where $D_{ij}$ is some smooth function. 
We additionally assume that our vector field is homogeneous and isotropic (in distribution), i.e.
 \begin{equation}\label{symmetries-spat}
 \vec f(\vec x;t)\stackrel{\law}{=}\vec f(\vec x+\vec a;t)
 \qquad\text{and}\qquad
 \vec f(\mathbf U\vec x;t)\stackrel{\law}{=}\mathbf U\vec f(\vec x;t)
 \end{equation}
for all translations $\vec a\in\R^n$ and all rotations $\mathbf U\in O(n)$.
It is worth noting that the dynamical equation~\eqref{SDE} preserves isotropy, while homogeneity is explicitly broken for $\mu>0$.

The statistical symmetries~\eqref{symmetries-spat} imply that the covariance tensor only depends on the spatial separation. More precisely, we have
\begin{equation}\label{covar-full}
D_{ij}(\vec r)\stackrel{\text{def}}{=}D_{ij}(\vec x,\vec x+\vec r)=\Gamma_1\Big(\frac{\norm{\vec r}^2}2\Big)\delta_{ij}+\Gamma_2\Big(\frac{\norm{\vec r}^2}2\Big)r_ir_j
\end{equation}
with $\Gamma_1$ and $\Gamma_2$ denoting smooth functions%
. We also require that
\begin{equation}\label{ineq}
\Gamma_1'(0)\leq0\leq\Gamma_1(0)
\qquad\text{and}\qquad
\frac{\Gamma_1'(0)}{n+1}\leq-\Gamma_2(0)\leq-\Gamma_1'(0).
\end{equation}
These latter constraints arises from obvious positivity conditions, see e.g.~\eqref{positivity} in the next section.

It worth making an additional remark before we continue with a stability analysis of our model. It is seen from~\eqref{SDE}, that for a fixed time $t>0$, we can apply an analysis identical to that produced by Fyodorov and Khoruzhenko~\cite{FK:2016} in order to find the expected number of equilibria at time $t$. However, this analysis would not tell us anything about stability. The reason for this disconnection between the number of equilibria and stability is clear: in our non-autonomous setting neither the location nor the number of equilibria is preserved in time. Thus, even for systems with a single stable equilibrium there is no guarantee for stability. 

\section{Lyapunov exponents and geometric Dyson Brownian motion}
\label{sec:from-DS-to-gDBm}

The main purpose of this section is to show how our stochastic dynamical system~\eqref{SDE} implies a stochastic description of the finite-time Lyapunov exponents. To do so, imagine that we introduce a small perturbation to the initial state of our system,  $\vec x(0)\mapsto\vec y(0)=\vec x(0)+\vec r(0)$. In order to study the Lyapunov exponents, we must study the evolution of the spatial separation $\vec r(t)=\vec y(t)-\vec x(t)$. Taking the integral equation~\eqref{SDE-int} as our starting point, we see that
\begin{equation}\label{SDE-r}
\vec r(t)-\vec r(s)=-\mu \int_s^t\vec r(u)du+\int_s^t \vec f(\vec y(u);u)du-\int_s^t \vec f(\vec x(u);u)du.
\end{equation}
For time-intervals in which the separation $\vec r(t)$ remains sufficiently small, we can use the expansion
\begin{equation}\label{taylor}
\int_s^t \vec f(\vec y(u);u)du=\int_s^t \vec f(\vec x(u);u)du+\int_s^t \mathbf J(\vec x(u);u)\vec r(u)du+\text{error},
\end{equation}
where (by construction) the Jacobian matrix $\mathbf J=(J_{ij})=(\p f_i/\p x_j)$ is zero-mean Gaussian and white-in-time. Using the regularity conditions imposed on $\vec f$, it can be seen that the error term in the expansion~\eqref{taylor} is subleading for sufficiently small separations. Thus, in the limit of an infinitesimal initial perturbation, the integral equation~\eqref{SDE-r} reduces to
\begin{equation}\label{SDE-r-int}
\vec r(t)-\vec r(s)=-\mu \int_s^t\vec r(u)du+\int_s^t \mathbf J(\vec x(u);u)\vec r(u)du.
\end{equation}
Writing the integral equation~\eqref{SDE-r-int} using the usual Stratonovich convention, we get
\begin{equation}\label{geo-vector}
d\vec r(t)=\big(-\mu\, dt\,\mathbf I+d\mathbf B(t)\big)\circ \vec r(t), \qquad
\mathbf B(t)\stackrel{\text{def}}{=}\int_0^t\mathbf J(\vec x(u);u)du.
\end{equation}
Here, $\mathbf B(t)$ is a zero-mean matrix-valued Brownian motion constructed such that the Jacobian matrix $\mathbf J(\vec x(t),t)$ is its formal time-derivative. 

Before we continue our study of the stochastic equation~\eqref{geo-vector}, we need to take a closer look at the noise term. It follows from the covariance tensor~\eqref{covar-full} that
\begin{align}\label{cor-brown}
\E[B_{ik}(t)B_{j\ell}(s)]&=\int_0^t\int_0^s\E[J_{ik}(\vec x(u);u)J_{j\ell}(\vec x(v);v)]dvdu \nn\\ 
&=-\big(\Gamma_1'(0)\delta_{ij}\delta_{k\ell}+\Gamma_2(0)(\delta_{ik}\delta_{j\ell}+\delta_{i\ell}\delta_{jk})\big)\min\{t,s\}.
\end{align}
We note that the fluctuations of the Brownian motion is independent of the spatial location $\vec x$. This spatial independence is a consequence of our vector field being both homogeneous-in-space and white-in-time; breaking either of these conditions will inevitable introduce spatial dependence, see section~\ref{sec:time}.
For notational simplicity, let us henceforth use abbreviations 
\begin{equation}
\Gamma_1'(0)=-\sigma^2
\qquad\text{and}\qquad
\Gamma_2(0)=-\sigma^2\tau.
\end{equation}
These constants must satisfy inequalities
\begin{equation}\label{inequal}
\sigma^2\geq 0 \qquad\text{and}\qquad -\frac1{n+1}\leq\tau\leq +1
\end{equation}
due to the positivity conditions
\begin{equation}\label{positivity}
\E\bigg[\bigg(\sum_{i=1}^nB_{ii}(t)\bigg)^{\!2\,}\bigg]\geq 0
\qquad\text{and}\qquad
\E\big[\big(B_{ij}(t)\pm B_{ji}(t)\big)^2\big]\geq0\quad \text{for}\ i\neq j.
\end{equation}
Note that the inequalities~\eqref{inequal} are consistent with the constraints~\eqref{ineq}. With the above notation, $\sigma^2$ has an interpretation as an overall `variance' of the Jacobian matrix, while $\tau$ is an interpolating parameter such that $\vec f$ is a conservative vector field for $\tau=+1$ and an incompressible vector field for $\tau=-1/(n+1)$. 

We can now return to our stochastic equation~\eqref{geo-vector}. Given an initial perturbation $\vec r(0)$, we may write the general solution of~\eqref{geo-vector} in the form $\vec r(t)=\mathbf U(t)\vec r(0)$. Here, the evolution matrix $\mathbf U(t)$ itself satisfies a stochastic equation
\begin{equation}\label{SDE-main}
d\mathbf U(t)=\big(-\mu\, dt\,\mathbf I+d\mathbf B(t)\big)\circ \mathbf U(t),\qquad\mathbf U(0)=\mathbf I.
\end{equation}
We may use this equation to find the Lyapunov spectrum. For completeness, let us briefly recall the meaning of Lyapunov exponents in our setting. We know from Osledec's multiplicative ergodic theorem~\cite{Oseledec:1968} that the following limit exists,
\begin{equation}
\lim_{t\to\infty}(\mathbf U(t)\mathbf U(t)^T)^{1/2t}=\mathbf H
\end{equation}
The matrix $\mathbf H$ is symmetric positive definite, hence we may denote its eigenvalues by $e^{\mu_1}\geq e^{\mu_2}\geq \cdots \geq e^{\mu_n}$. The exponents $\mu_1\geq \mu_2\geq \cdots\geq \mu_n$ are (by definition) called the Lyapunov exponents. If the largest Lyapunov exponent is unique ($\mu_1>\mu_2$) then we have the asymptotic property
\begin{equation}
\norm{\vec r(t)}\sim e^{\,\mu_1t}\norm{\vec r(0)} \qquad \text{for} \qquad t\to\infty.
\end{equation}
In other words, we have the familiar property that two points which are initially close tends to separate exponentially fast if the largest Lyapunov exponent is greater than zero, while their separation tends to decay exponentially fast when the Lyapunov spectrum is purely negative. 

The Lyapunov exponents are invaluable tool for studies of asymptotic behaviour, but they neglect any finite-time effects which may be physically relevant (see e.g.~\cite{MS:2014}). Amazingly, it turns out that our model is simple enough to allow us to construct a stochastic description for the finite-time Lyapunov exponents, which includes all finite-time effects. Let $e^{\lambda_1(t)}\geq\cdots\geq e^{\lambda_n(t)}$ denote the singular values of the evolution matrix $\mathbf U(t)$, then we define the finite-time Lyapunov exponents as $\lambda_1(t)\geq \cdots\geq \lambda_n(t)$. We note that the finite-time Lyapunov exponents are defined such that
\begin{equation}\label{lyapunov-def}
\lambda_i(t)/t\to\mu_i \qquad\text{for}\qquad t\to\infty.
\end{equation}
The stochastic equation~\eqref{SDE-main} is a matrix-valued generalisation of a geometric Brownian motion and it follows from Theorem~\ref{thm:geo-dyson-brown} in Appendix~\ref{appendix} that the finite-time Lyapunov exponents are given as a geometric Dyson Brownian motion,
\begin{equation}\label{SDE-lambda}
d\lambda_i(t)=(1+2\tau)^{1/2}\sigma dB_{i}(t)-\mu dt+\frac{(1+\tau)\sigma^2}2\sum_{j=1,j\neq i}^n\frac{dt}{\tanh(\lambda_i(t)-\lambda_j(t))},
\qquad i=1,\ldots,n.
\end{equation}
Furthermore, let $\rho(\bm\lambda;t)$ denote the joint probability density function for the finite-time Lyapunov exponents, then it follows from Corollary~\ref{cor:dyson}, that we have a Fokker--Planck--Kolmogorov equation given by
\begin{equation}\label{PDF-lambda}
\frac{\p\rho(\bm\lambda;t)}{\p t}=\sum_{i=1}^n\bigg(\mu\frac{\p\rho(\bm\lambda;t)}{\p\lambda_i}
-\frac{(1+\tau)\sigma^2}2\sum_{j=1,j\neq i}^n\frac{\p}{\p\lambda_i}\frac{\rho(\bm\lambda;t)}{\tanh(\lambda_i-\lambda_j)}
+\frac{(1+2\tau)\sigma^2}{2}\frac{\p^2\rho(\bm\lambda;t)}{\p\lambda_i^2}\bigg)
\end{equation}
with initial conditions $\lambda_1(0)=\cdots=\lambda_n(0)=0$.
These two equations give us a full description of the finite-time Lyapunov exponents at any time $t>0$. Figure~\ref{fig} shows a typical realisation of the finite-time Lyapunov exponents as described by the stochastic equations~\eqref{SDE-lambda}. We see that the finite-time Lyapunov exponents appear to diverge linearly in the long time limit. 
In the next section, we will look at the asymptotic behaviour of the finite-time Lyapunov exponents and use this to draw a stability-complexity phase diagrams. 
\begin{figure}[htbp]
\centering
 \includegraphics{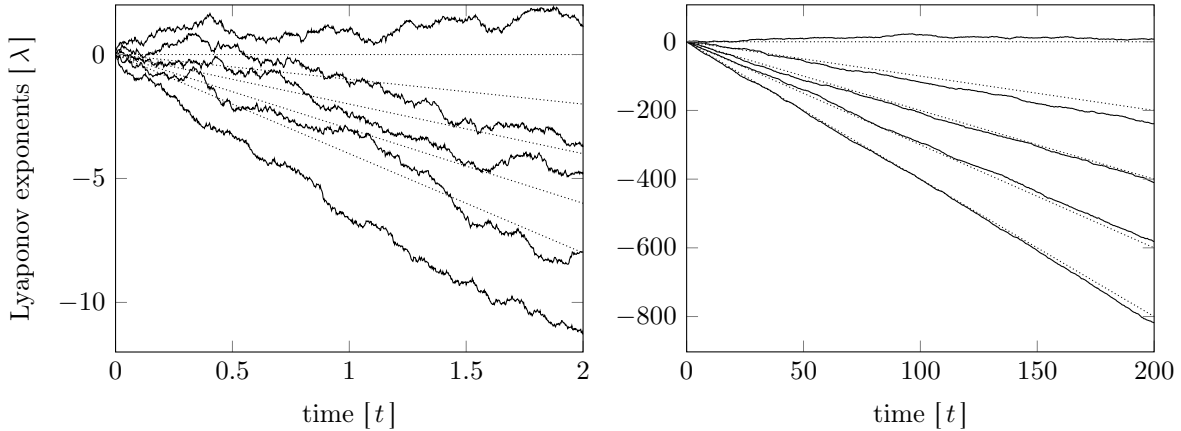}
\caption{The figure shows realisations of the geometric Dyson Brownian motion~\eqref{SDE-lambda} on time scales $t\in[0,2]$ (left panel) and $t\in[0,200]$ (right panel). Paths on both panels are for the parameter choice $(n,\mu,\sigma,\tau)=(5,2,1,0)$, and the paths themselves are constructed using $1\,000\,000$ time steps ($1\,000$ of these are used for the plots). The dotted lines shows the linear part of the large-$t$ asymptotic behaviour~\eqref{asymp}.}\label{fig}
\end{figure}

\section{May--Wigner transition and the weak-noise limit}
\label{sec:stability}

In this section we will use the stochastic description of the finite-time Lyapunov exponents~\eqref{SDE-lambda} and~\eqref{PDF-lambda} to describe stability of our system. The first step is to look at the long-time behaviour of the finite-time Lyapunov exponents, i.e. Lyapunov exponents themselves. The Lyapunov exponents~\eqref{lyapunov-def} can be obtained rigorously using several different approaches~\cite{LeJan:1985,Baxendale:1986,Newman:1986,NRW:1986} and can (up to a shift $\mu$) be found in~\cite{LeJan:1985} (see Proposition~1 and the remark below). We see no reason to repeat any of these methods here, rather let us provide a heuristic (but physically illuminative) derivation. 

Our heuristic derivation is based on the justified (see e.g.~\cite{Newman:1986}) assumption that the Lyapunov spectrum is non-degenerate, $\mu_1>\cdots>\mu_n$. Consequently, the finite-time Lyapunov exponents diverge as
\begin{equation}\label{lya-sep}
\abs{\lambda_i(t)-\lambda_j(t)}\sim\abs{\mu_i-\mu_j}t
\qquad \text{for} \quad t\to\infty \quad \text{and} \quad i\neq j.
\end{equation}
To proceed, we pick two times $t>s>0$ and assume that the finite-time Lyapunov exponents at time $s$ are known and pairwise distinct. The evolution of the exponents from time $s$ to time $t$ is given by our stochastic equations~\eqref{SDE-lambda},
\begin{equation}
\lambda_i(t)=\lambda_i(s)+(1+2\tau)^{1/2}\sigma (B_{i}(t)-B_{i}(s))-\mu (t-s)
+\frac{(1+\tau)\sigma^2}2\sum_{j=1,j\neq i}^n\int_s^t\frac{du}{\tanh(\lambda_i(u)-\lambda_j(u))}.
\end{equation}
Now, assuming that $s\gg1$ and using the asymptotic formulae~\eqref{lya-sep} together with the ordering $\mu_1>\cdots>\mu_n$ gives
\begin{equation}
\int_s^t\frac{du}{\tanh(\lambda_i(u)-\lambda_j(u))}\sim
\begin{cases}t-s & \quad\text{for}\quad i<j \\ s-t & \quad\text{for}\quad i>j\end{cases}
\end{equation}
with exponentially suppressed correction terms. Thus, the asymptotic behaviour of the finite-time Lyapunov exponents is given by
\begin{equation}\label{asymp}
\lambda_i(t)\sim\lambda_i(s)+(1+2\tau)^{1/2}\sigma (B_{i}(t)-B_{i}(s))+\mu_i(t-s)
\end{equation}
with
\begin{equation}\label{lya-value}
\mu_i=\frac{(1+\tau)\sigma^2(n-2i+1)}2-\mu.
\end{equation}
We note that the asymptotic formula~\eqref{asymp} is valid for any $n\geq1$.

The approximation scheme given by~\eqref{asymp} and~\eqref{lya-value} contains important information about our system. We see that at short times the finite-time Lyapunov exponents are strongly interacting (they repel), while in the asymptotic limit they become independent Brownian motions each with their own drift term $\mu_i$. We note that $\lambda_i(t)/t\to\mu_i$ for $t\to\infty$, hence the drift terms are the Lyapunov exponents. A numerical verification of this phenomenon is seen on figure~\ref{fig}.

Knowledge about the fluctuations of the largest Lyapunov exponent is important for studies of the May--Wigner transition (see e.g.~\cite{MS:2014}). However, in this paper we will limit ourselves to finding the location of the phase transition (i.e. we determine the phase diagram) and leave a detailed study of the phase transition itself for future work. Our system is said to be stable (unstable) if the largest Lyapunov exponent is less (greater) than zero, $\mu_1<0$ ($\mu_1>0$). Thus, it is immediate from~\eqref{lya-value} that our system is stable if $(1+\tau)(n-1)/2<\mu/\sigma^2$ and unstable otherwise. It is usual to introduce rescaled units $\hat\mu=\mu/\sigma^2n$, such that in the large-$n$ limit the May--Wigner transition happens at the critical value $\hat\mu_c=(1+\tau)/2$. It is worth to compare this result with the linear model~\eqref{linear} studied in~\cite{IS:2016}. We note that the two models predict a phase transition at the same location, namely $\hat\mu_c=(1+\tau)/2$. However, there is a crucial difference between the two models: the linear models allow an interpolation parameter $-1\leq\tau\leq+1$ (see~\cite{IS:2016}), while our present model restrict parameter space to $0\leq \tau\leq +1$. We note that the difference between the two models boils down to whether we first do a linearisation and then impose symmetry constraints (as in~\cite{IS:2016}) or we first impose symmetry constraints and then linearise (as in this paper). 

The above given description of the phase diagram rests upon the assumption that the Lyapunov spectrum is non-degenerate. This is true for $\sigma>0$ but fails for $\sigma=0$ where all Lyapunov exponents are identical and equal to $\mu$. This suggest that there is special behaviour in the weak-noise limit $\sigma\to0$. 

Formally the weak-noise limit is obtained by introducing $\sigma=\epsilon\tilde\sigma$ and $\lambda_i(t)=\epsilon\tilde\lambda_i(t)$ followed by the limit $\epsilon\to0$. With this rescaling, the geometric Dyson Brownian motion~\eqref{SDE-lambda} reduces to
\begin{equation}\label{dyson-process}
d\tilde\lambda_i(t)=(1+2\tau)^{1/2}\tilde\sigma dB_{i}(t)-\mu dt
+\frac{(1+\tau)\tilde\sigma^2}2\sum_{j=1,j\neq i}^n\frac{dt}{\tilde\lambda_i(t)-\tilde\lambda_j(t)},
\qquad i=1,\ldots,n.
\end{equation}
Likewise, the Fokker--Planck--Kolmogorov equation becomes
\begin{equation}
\frac{\p\tilde\rho(\bm{\tilde\lambda};t)}{\p t}=\sum_{i=1}^n\bigg(\mu\frac{\p\tilde\rho(\bm{\tilde\lambda};t)}{\p\lambda_i}
-\frac{(1+\tau)\sigma^2}2\sum_{j=1,j\neq i}^n
\frac{\p}{\p\tilde\lambda_i}\frac{\tilde\rho(\bm{\tilde\lambda};t)}{\tilde\lambda_i-\tilde\lambda_j}
+\frac{(1+2\tau)\sigma^2}{2}\frac{\p^2\tilde\rho(\bm{\tilde\lambda};t)}{\p\tilde\lambda_i^2}\bigg).
\end{equation}
We recognise this as an ordinary Dyson Brownian motion~\cite{Dyson:1962brown} (see~\cite{Katori:2016} for a review). This stochastic system has a completely different asymptotic behaviour. Most importantly, the interaction between finite-time Lyapunov exponents does not die out in the long-time regime.

Let us consider the stability of the weak-noise system~\eqref{dyson-process}. In this case the Lyapunov spectrum becomes trivial $\mu_1=\cdots=\mu_n=\mu$. This is merely a consequence of the finite-time Lyapunov exponents separating slower than linearly with time, $\abs{\tilde\lambda_i(t)-\tilde\lambda_j(t)}=o(t)$ for $t\to\infty$. Thus, we cannot use the Lyapunov exponents to say anything about stability. However, we have more knowledge than the Lyapunov exponents; we know that our system is stable if all the finite-time Lyapunov exponents become  negative as time tends to infinity. It is well-known (see e.g.~\cite{Katori:2016}) that the global spectrum of the Dyson process~\eqref{dyson-process} is a growing semi-circle with drift and that the largest finite-time Lyapunov exponent grows as
\begin{equation}
\tilde\lambda_1(t)\sim\tilde\sigma\sqrt{(1+\tau)nt}-\tilde\sigma n^{1/2}\tilde\mu t+o(n^{1/2})
\qquad\text{for}\qquad t\to\infty,
\end{equation}
where we have introduced the convenient scale $\tilde\mu=\mu/\tilde\sigma n^{1/2}$. We see that the critical value for stability is $\tilde\mu_c=((1+\tau)/t)^{1/2}$ which tends to zero as time goes to infinity. We conclude that any weak-noise system (independently of the system size) with a positive decay rate $\mu$ is asymptotically stable. However, from a physical perspective it worth noting that the relevant time-scale for stability is set by $t\gg\tilde\sigma n^{1/2}$, hence we might need to wait a while.

\section{Time-effects and persistent auto-correlations}
\label{sec:time}

Let us return to our dynamical system~\eqref{SDE} described in Section~\ref{sec:model}. In this section, we look at the effect of weakening the condition that the correlations of our vector field must be white-in-time. More precisely, we will replace the correlation tensor~\eqref{covar} with
\begin{equation}\label{covar-time}
\E[f_i(\vec x;t)f_j(\vec x+\vec r,s)]=D_{ij}(\vec r)g(\abs{t-s}),\qquad i,j=1,\ldots,n,
\end{equation}
where $D_{ij}(\vec r)$ is given by~\eqref{covar-full} and $g$ is a smooth monotonically decreasing function. This new dynamical system has the same symmetries as our original system, i.e. the vector field is homogeneous and isotropic in space as well as stationary and time-reversible. However, we will see below that the non-trivial time-dependence given by the function $g$ might introduce persistent time-correlations and thereby have fundamental consequences for the dynamics.

We will proceed as in Section~\ref{sec:from-DS-to-gDBm} and investigate evolution of a small initial perturbation. Following the same steps as in Section~\ref{sec:from-DS-to-gDBm}, we see that the perturbation $\vec r(t)$ once again satisfies an integral equation
\begin{equation}
\vec r(t)-\vec r(s)=-\mu \int_s^t\vec r(u)du+\int_s^t \mathbf J(\vec x(u);u)\vec r(u)du,
\end{equation}
where the Jacobian matrix $(J_{ij})=(\p f_i/\p x_j)$ is a Gaussian noise matrix. The Jacobian correlation tensor reads
\begin{equation}
\E[J_{ik}(\vec x;t)J_{j\ell}(\vec y;s)]=C_{ij;k\ell}(\vec x-\vec y)g(\abs{t-s})
\end{equation}
with $C_{ij;k\ell}$ being a fourth-order tensor given by differentiation of~\eqref{covar-time}. Explicitly, we have
\begin{multline}\label{C-tensor}
C_{ij;k\ell}(\vec r)=-\Gamma_1'\Big(\frac{\norm{\vec r}^2}2\Big)\delta_{ij}\delta_{k\ell}
-\Gamma_2\Big(\frac{\norm{\vec r}^2}2\Big)(\delta_{ik}\delta_{j\ell}+\delta_{i\ell}\delta_{jk})
-\Gamma_1''\Big(\frac{\norm{\vec r}^2}2\Big)\delta_{ij}r_kr_\ell \\
-\Gamma_2'\Big(\frac{\norm{\vec r}^2}2\Big)
(\delta_{ik}r_jr_\ell+\delta_{jk}r_ir_\ell+\delta_{i\ell}r_jr_k+\delta_{j\ell}r_ir_k+\delta_{k\ell}r_ir_j)
-\Gamma_2''\Big(\frac{\norm{\vec r}^2}2\Big)r_ir_jr_kr_\ell.
\end{multline}
As in Section~\ref{sec:from-DS-to-gDBm}, rather than working with the Jacobian matrix itself, it is convenient to introduce the integrated process 
\begin{equation}\label{int-process}
\mathbf B(\vec x(t);t)\stackrel{\text{def}}{=}\int_0^t\mathbf J(\vec x(u);u)du.
\end{equation}
By construction, this integrated process is zero-mean Gaussian with auto-correlations 
\begin{equation}\label{cor-int-tensor}
\E[B_{ik}(\vec x(t);t)B_{j\ell}(\vec x(s);s)]
=\int_0^t\int_0^sC_{ij;k\ell}(\vec x(u)-\vec x(v))g(\abs{u-v})dvdu.
\end{equation}
We see that in the white-in-time limit, $g(t)\to\delta(t)$, the correlations~\eqref{cor-int-tensor} with~\eqref{C-tensor} for the integrated process~\eqref{int-process} reduces to the matrix-valued Brownian motion~\eqref{cor-brown}. Thus, the integrated process becomes independent of the spatial location and we reclaim our previous results (as we should). 

Intuitively, we expect that if the function $g(t)$ is sharply peaked about the origin then we may use the white-in-time limit as an approximation. In order to put a physical handle on this phenomenon, let us introduce the concept of an auto-correlation time for a trajectory. Consider a trajectory $\vec x(t)$ given as the solution to the differential equation~\eqref{SDE} with initial condition $\vec x(0)=\vec x$. We define the corresponding auto-correlation time as
\begin{equation}\label{auto-time}
T(\vec x)\stackrel{\text{def}}{=}
\lim_{t\to\infty}\int_0^t \frac{\E[\vec f(\vec x;0)\bm\cdot\vec f(\vec x(u);u)]}{\E[\norm{\vec f(\vec x;0)}^2]} du
=\frac{1}{\Gamma_1(0)g(0)}\lim_{t\to\infty}\int_0^t \Gamma_1\bigg(\frac{\norm{\vec x(u)-\vec x}^2}2\bigg)g(\abs{u}) du.
\end{equation}
The auto-correlation time~\eqref{auto-time} may be considered as a measure of `trajectory memory' for a trajectory starting at $\vec x$. 
We note that spatial homogeneity in our dynamical system~\eqref{SDE} is explicitly broken for $\mu>0$, thus trajectories and therefore the auto-correlation time~\eqref{auto-time} will in general depend on the starting point $\vec x$. 

Let us take a closer look at the auto-correlation time~\eqref{auto-time}. We know from the constraints~\eqref{ineq} that $\Gamma_1'(0)\leq0\leq\Gamma_1(0)$. In other words, the function $\Gamma_1$ is decreasing in a neighbourhood of zero.
From a physical perspective, it is natural to assume that this extends to a global property, i.e. spatial correlations decays with separation.
With this assumption, we have the upper bound $\Gamma_1(x)\leq \Gamma_1(0)$ which ensures that the auto-correlation time is finite throughout space as long as $g$ decays sufficiently fast. The white-in-time approximation is then reasonable on time-scales $t\gg T$. On the other hand, if $g$ decays too slowly then the auto-correlation time might be divergent, which indicates persistent time-correlations. The worst case scenario is when $g$ is constant, which corresponds to the Fyodorov--Khoruzhenko model~\cite{FK:2016}. In this scenario, it is evident from~\eqref{auto-time} that the auto-correlation time will diverge if the trajectory remains in a small region of space for a long time, which is the case for trajectories starting in the neighbourhood of an equilibrium. 

\section{Conclusions and open problems}
\label{sec:conclusion}

In this paper we have introduced a generic model for study of stability in large complex systems. Our model may be considered a non-autonomous version of a model recently introduced by Fyodorov and Khoruzhenko~\cite{FK:2016}. However, the behaviour of our model differs considerably from the behaviour of the Fyodorov--Khoruzhenko model due to its explicit time-dependence. Most importantly, we have seen that the stability-complexity phase transition (or May--Wigner transition) happens at a scale proportional to the system size $n$, while the same transition happens at scale proportional to $n^{1/2}$ in the Fyodorov--Khoruzhenko model. These types of scaling relations are important as they present one of the most promising avenues towards experimental tests. We note that a recent paper~\cite{BR:2016} claims to have experimentally confirmed the scaling $n^{1/2}$ from a closely related model~\cite{SCS:1988} occurring in the description of neural networks.

From a more theoretical perspective, we have seen that it is possible to construct a stochastic description of the finite-time Lyapunov exponents for spatial homogeneous and isotropic dynamical systems under the assumption that the temporal fluctuations are sufficiently rapid. If the temporal fluctuations are too slow then the stochastic description breaks down and trajectories with persistent time-correlations become important for the dynamics. 

Moreover, we have shown that there is a close connection between the non-linear dynamical system studied in this paper and the linear model~\eqref{linear} studied in~\cite{IS:2016}. That is a connection to matrix-valued stochastic equations of multiplicative type and geometric Dyson Brownian motions. In our opinion, these models are by themselves interesting. For instance, it would be intriguing to study the global spectrum of the finite-time Lyapunov exponents, which are related to the theory of free stochastic equation. This is of particularly interest since the short-time behaviour of the finite-time Lyapunov exponents must be very different from the asymptotic behaviour. It would also be worthwhile to study local correlations; in particular the largest exponent which determines the behaviour at the phase transition (see e.g.~\cite{MS:2014} where fluctuations of of the largest Lyapunov exponent was used to show that the May--Wigner transition for symmetric Gaussian community matrices is a third-order phase transtion).
Another problem worth pursuing is a study of discrete time models, which would relate to products of random matrices~\cite{CPV:1993,AI:2015}.

\paragraph{Acknowledgement:} We thank Peter Forrester for comments on a first draft of this paper, Yan Fyodorov for sharing an early version of the paper~\cite{BAFK:2017}, and Andre Peterson for bringing the paper~\cite{BR:2016} to our attention. The author acknowledge financial support by ARC Centre of Excellence for Mathematical and Statistical frontiers (ACEMS).

\appendix

\section{Geometric Dyson Brownian motion}
\label{appendix}

Let $\mathbf B(t)$ be a $\mathfrak{gl}_n(\R)$-valued isotropic Brownian motion, i.e. $\mathbf B(t)=[B_{ij}(t)]$ is an $n$-by-$n$ matrix those entries are Brownian motions such that~\cite{LeJan:1985,IS:2016}
\begin{equation}\label{brown-corr}
\E\big[B_{ij}(t)\big]=0 \qquad \text{and} \qquad
\E\big[B_{ik}(t)B_{j\ell}(s)]=(a\delta_{ij}\delta_{k\ell}+b\delta_{ik}\delta_{j\ell}+c\delta_{i\ell}\delta_{jk})\min\{t,s\}
\end{equation}
with $a\pm c\geq 0$ and $a+nb+c\geq 0$. Consider the matrix-valued generalisation of a geometric Brownian motion,
\begin{equation}\label{geo-brown}
d\mathbf U(t)=\big(d\mathbf B(t)-\mu dt\mathbf I\big)\circ \mathbf U(t)
=\big(d\mathbf B(t)+(\tfrac{a+b+nc}{2}-\mu)dt\mathbf I\big)\mathbf U(t),\qquad \mathbf U(0)=\mathbf I,
\end{equation}
where $\mathbf I$ denotes the identity matrix and $\mathbf U(t)$ is a stochastic process which is $\textup{GL}_n(\R)$-valued (a.s.) in virtue of the Stroock--Varadhan support theorem. 

Let $e^{\lambda_1(t)},\ldots,e^{\lambda_n(t)}$ denote the singular values of $\mathbf U(t)$ at time $t>1$. We define the finite-time Lyapunov exponents as ${\lambda_1(t)},\ldots,{\lambda_n(t)}$ and we observe that these exponents are finite and pairwise distinct for any fixed time $t$ (a.s.) unless $a+c=b=0$. For $a+c=b=0$, the Brownian motion $\mathbf B(t)$ is anti-symmetric, hence $\mathbf U(t)$ belongs to orthogonal group (a.s.). Consequently, all the finite-time Lyapunov exponents are trivially equal to zero in this case.

\begin{theorem}\label{thm:geo-dyson-brown}
Let $\mathbf U(t)$ be a matrix-valued geometric Brownian motion defined as the solution to~\eqref{geo-brown}, then the corresponding finite-time Lyapunov exponents are solutions to the coupled stochastic equations
\begin{equation}\label{geo-dyson-brown}
d\lambda_i(t)=(a+b+c)^{1/2}dB_{i}(t)-\mu dt+\frac{a+c}2\sum_{j=1,j\neq i}^n\frac{dt}{\tanh(\lambda_i(t)-\lambda_j(t))}
\qquad i=1,\ldots,n
\end{equation}
with $B_1(t),\ldots,B_n(t)$ being independent standard Brownian motions. 
\end{theorem}

\begin{corollary}\label{cor:dyson}
Let $\rho(\,\cdot\,;t)$ denote the joint probability density of the finite-time Lyapunov exponents at time $t>0$, then it satisfies the Fokker--Planck--Kolmogorov equation
\begin{equation}\label{dyson-diff}
\frac{\p}{\p t}\rho(\lambda_1,\ldots,\lambda_n;t)=L\rho(\lambda_1,\ldots,\lambda_n;t),
\end{equation}
where
\begin{equation}\label{dyson-diff-oper}
L=\sum_{i=1}^n\bigg(\mu\frac{\p}{\p\lambda_i}
-\frac{a+c}2\sum_{j=1,j\neq i}^n\frac{\p}{\p\lambda_i}\frac{1}{\tanh(\lambda_i-\lambda_j)}
+\frac{a+b+c}{2}\frac{\p^2}{\p\lambda_i^2}\bigg).
\end{equation}
\end{corollary}

\begin{remark}
We will refer to~\eqref{geo-dyson-brown} as a geometric Dyson Brownian motion.
Note that a geometric Dyson Brownian motion with $b=c=\mu=0$ was considered in~\cite{NRW:1986}, while we considered the case $b=0$ in~\cite{IS:2016} where also a complex-valued version was introduced and studied. Actually, the Fokker--Planck--Kolmogorov equation~\eqref{dyson-diff} is not explicitly mentioned in~\cite{NRW:1986}, while the stochastic formulation~\eqref{geo-dyson-brown} is not explicitly mentioned in~\cite{IS:2016}, but these two formulations are naturally closely related.
\end{remark}

\begin{proof}[Proof of Theorem~\ref{thm:geo-dyson-brown}]
We are interested in the finite-time Lyapunov exponents which relates to the singular values of $\mathbf U(t)$, thus it is convenient to introduce a Wishart-type matrix
\begin{equation}
\mathbf S(t)=\mathbf U(t)\mathbf U(t)^T.
\end{equation}
For a small time increment $dt>0$, it follows from~\eqref{geo-brown} that 
\begin{equation}\label{S-increment}
d\mathbf S(t)=d\mathbf B(t)\mathbf S(t)+\mathbf S(t)d\mathbf B(t)^T+d\mathbf B(t)\mathbf S(t)d\mathbf B(t)^T
+(a+b+nc-2\mu) dt\mathbf S(t)+O(dt^{3/2})
\end{equation}
As $\mathbf S(t)$ is a symmetric matrix, we may proceed as for ordinary Dyson Brownian motion~\cite{Dyson:1962brown} and use perturbation theory. If $s_1(t),\ldots,s_n(t)$ denote the eigenvalues of $\mathbf S(t)$, then we have
\begin{equation}\label{perturb}
ds_i(t)=dS_{ii}(t)+\sum_{j=1,j\neq i}^n\frac{dS_{ij}(t)dS_{ji}(t)}{s_i(t)-s_j(t)}+O(dt^{3/2}), \qquad i=1,\ldots,n.
\end{equation}
We recall that the eigenvalues $s_1(t),\ldots,s_n(t)$ are pairwise distinct (a.s.). Furthermore, it follows from~\eqref{S-increment} and~\eqref{brown-corr} that
\begin{equation}\label{dSii}
dS_{ii}(t)=2s_i(t)dB_{ii}(t)+a\,dt\sum_{j=1}^ns_j(t)+(a+2b+(n+1)c-2\mu) dt\, s_i(t)+O(dt^{3/2})
\end{equation}
and
\begin{equation}\label{dSij}
dS_{ij}(t)dS_{ji}(t)=a(s_i(t)^2+s_j(t)^2)dt+2c\, s_i(t)s_j(t)dt+O(dt^{3/2})  \qquad \text{for }i\neq j.
\end{equation}
Inserting~\eqref{dSii} and~\eqref{dSij} into the perturbative expansion~\eqref{perturb} gives
\begin{align}
ds_i(t)=2s_i(t)dB_{ii}(t)&+a\,dt\sum_{j=1}^ns_j(t)+(a+2b+(n+1)c-2\mu) dt\, s_i(t) \nn \\
&+dt\sum_{j=1,j\neq i}^n\frac{a(s_i(t)^2+s_j(t)^2)+2c\, s_i(t)s_j(t)}{s_i(t)-s_j(t)}+O(dt^{3/2}),
\end{align}
which after some standard manipulations of the sums may be rewritten as
\begin{equation}
\frac{ds_i(t)}{2s_i(t)}=dB_{ii}(t)+(a+b+c-\mu)dt+\frac{a+c}2dt\sum_{j=1,j\neq i}^n\frac{s_i(t)+s_j(t)}{s_i(t)-s_j(t)}+O(dt^{3/2}).
\end{equation}
It remains to change variables from the eigenvalues $s_1(t),\ldots,s_n(t)$ to the finite-time Lyapunov exponents $\lambda_1(t),\ldots,\lambda_n(t)$. By construction, we have $s_i(t)=e^{2\lambda_i(t)}$, thus it follows by It\^o's formula that
\begin{equation}
\frac{ds_i(t)}{2s_i(t)}=d\lambda_i(t)+(a+b+c)dt+O(dt^{3/2})
\end{equation}
and consequently that
\begin{equation}
d\lambda_i(t)=dB_{ii}(t)-\mu dt+\frac{a+c}2\sum_{j=1,j\neq i}^n\frac{dt}{\tanh(\lambda_i(t)-\lambda_j(t))}+O(dt^{3/2}).
\end{equation}
The final observation is that due to~\eqref{brown-corr}, we may make the replacement $B_{ii}(t)\mapsto(a+b+c)^{1/2}B_i(t)$ where $B_1(t),\ldots,B_n(t)$ are independent standard Brownian motions.
\end{proof}

Corollary~\ref{cor:dyson} follows from the Theorem~\ref{thm:geo-dyson-brown} using standard arguments, see e.g.~\cite{Tao:2012}. Perhaps the easiest approach is to first introduce a test function $\phi:\R^n\to \R$ which is smooth with bounded derivatives. By means of Taylor expansion, it is straightforward to show that for any such test function, we have
\begin{equation}
\frac{\p}{\p t}\E[\phi(\lambda_1,\ldots,\lambda_n)]=\E[L^*\phi(\lambda_1,\ldots,\lambda_n)],
\end{equation}
where $L^*$ is the adjoint geometric Dyson operator, i.e.
\begin{equation}
L^*\phi=\sum_{i=1}^n\bigg(-\mu\frac{\p\phi}{\p\lambda_i}+\frac{a+c}2\sum_{j=1,j\neq i}^n\frac{1}{\tanh(\lambda_i-\lambda_j)}\frac{\p\phi}{\p\lambda_i}
+\frac{a+b+c}{2}\frac{\p^2\phi}{\p\lambda_i^2}\bigg).
\end{equation}
The corollary then follows using integration by parts.




\end{document}